\documentclass{article}
\usepackage{spconf,amsmath,graphicx}
\usepackage{amssymb}
\usepackage{bm}
\usepackage{amsthm}

\newtheorem{Theorem}{\textbf{Theorem}}

\newtheorem{Corollary}{\textbf{Corollary}}
\title{Performance Improvement of the Compressive Classifier Using Equi-Norm Tight Frames}
\name{Hailong Shi , Hao Zhang and Xiqin Wang
\thanks{This work was supported in part by the National Basic Research Program of China (973 Program, No. 2010CB731901) and in part by the National Natural Science Foundation of China (No. 40901157 and No. 61201356).}}
\address{Intelligent Sensing Lab, Department of Electronic Engineering, Tsinghua University, Beijing\\
Email: shl06@mails.tsinghua.edu.cn}
\begin{document}
%
\maketitle
\begin{abstract}
\par Classifying sparse signals contaminated by Gaussian noise with compressive measurements is different from common sparse recovery, for its focus is minimizing the probability of false classification instead of error of estimation. This paper considered the way to decrease the probability of false classification for compressive classifier. It is proved rigorously that the probability of false classification could be reduced by tighten the measurement matrix, that is, to make it row-orthogonal simply. As is well known, equiangular tight frame (ETF) is the best choice for measurement matrix for its optimal worst-case performance. But its existence and construction is problematic on not a few dimensions. So our results provided a convenient approach to improve the performance of compressive classifiers -- The tightened measurement matrices could be better than before. Numerical results illustrated the validity of our conclusion.
\end{abstract}
\section{Introduction}
\label{sec:intro}

\par In recent years considerable progress has been made towards sensing and recovery of sparse signals, known as Compressive Sensing \cite{CSCandes}\cite{RobustUncertainty}. However, less attention has been paid to classifying sparse signals from compressive measurements, which can be stated as Compressive Classification.
\par The Compressive Classification is derived from the theory of statistical signal processing \cite{KayV2}. In this compressive scenario, different sparse signals contaminated by gaussian noise are compressed by under-determined linear transformations, and then the compressed measurements are used for signal classification. Since the objective in classification is to minimize the error probability, or the probability of false classification of the classifiers, analysis of the influence of measurement matrices on the classification is different from that of estimating sparse signals.
\par To the best of our knowledge, the earliest analysis of the performance of Compressive Classification is from \cite{Davenport06detectionand}\cite{SPComp} by Davenport in 2006. In these publications, performance bounds on the probability of false classification of Compressive Classifiers using row-orthogonal random projection matrices were analyzed.
Besides, there are other publications such as Ramin Zahedi \textit{et al.} 's works \cite{Ramin2010}\cite{Zahedi201264}, in which Equiangular Tight Frames (ETFs, \cite{Waldron2009}) were considered and proved to have best worst-case performance as the measurement matrices in Compressive Classification under the constraint that the measurement matrices are row-orthogonal.
\par
This paper is mainly concerned about the implementation of measurement matrices to improve the performance of Compressive Classification. We noticed that all these publications mentioned above (\cite{Davenport06detectionand}\cite{SPComp}\cite{Ramin2010}\cite{Zahedi201264}) and many other related most recent researches (\cite{RaoICCASP2012}\cite{PerfLimitsTranSP}) did their analysis under the constraint that the measurement matrices are row-orthogonal, \textit{i.e.} the column vectors form an equi-norm tight frame. The reason why they use this constraint is for simplicity and avoiding the coloring of received signals between different hypotheses. While in this paper, we will prove that the transformation to equi-norm tight frames from arbitrary measurement matrices will reduce the probability of false classification of the commonly used Compressive Classifier, which coincides with the row-orthogonal constraint commonly used above in (\cite{Davenport06detectionand}\cite{SPComp}\cite{Ramin2010}\cite{Zahedi201264}). Although there are already proofs (\cite{Ramin2010}\cite{Zahedi201264}) showing that Equiangular Tight Frames (ETFs) have best worst-case performance among all tight frames, our job is different from theirs. Because constructing equi-norm tight frames is much simple and practical, so our results could provide a convenient approach to improve the performance of Compressive Classifiers -- The tightened, or row-orthogonalized measurement matrices could be better than before.
\par
The remainder of this paper is organized as follows:
In Section 2, we give the problem formulation of Compressive Classification and the error probability; in Section 3, we demonstrate our main result and proof of the influence of equi-norm tight frames on the error probability, and then some discussions and remarks; and in the last section numerical simulation results and discussions are provided for validation.

\section{Compressive Signal Classification}
\label{sec:compclassf}
\par
Derived from the traditional theory of statistical signal detection and estimation \cite{KayV2}, in the Compressive Classification scenario,  the sparse signals is compressed by under-determined linear transformation first, then these compressed measurements are used to do classification. As is stated in Davenport's works \cite{Davenport06detectionand}\cite{SPComp}, the model of the 2-ary compressive signal classification can be expressed as:
\begin{equation}\label{CompressHypo}
\bm{y} =
\left \{ \begin{array}{ll}
	\bm{\Phi}(\bm{s_1} + \bm{n}) & \text{Hypothesis  $H_1$} \\
	\bm{\Phi}(\bm{s_2} + \bm{n}) & \text{Hypothesis  $H_2$} \\

\end{array} \right.
\end{equation}
in which $\bm{s_i}\in\Lambda_k=\{\bm{s} \in\mathbb{R}^N,\|\bm{s}\|_0\leq k\},i=1,2$ are k-sparse signals, i.e. the vectors with at most k nonzero elements, besides they are orthogonal to each other and have equal norms ($\langle \bm{s_i},\bm{s_j}\rangle = 0,\|\bm s_i\|_2^2=\|\bm s_j\|_2^2,i \neq j$). $\bm{y}\in \mathbb{R}^n$ is the received compressive measurements, $\bm{n}\thicksim \mathcal{N} (0,\sigma^2\bm{I_N})$ is the Gaussian noise contaminating the sparse signals. And $\bm{\Phi}\in \mathbb{R}^{n\times N},n<N$ is the under-determined measurement matrix, with full row rank and equi-norm columns, and satisfying Restricted Isometry Property (RIP) \cite{Decoding}\cite{SimpleProof}. It is best that we focus on the scenario of 2-ary Compressive Classification first, which means there are 2 hypotheses in (\ref{CompressHypo}). More general results of multiple hypotheses (more than 2) can be easily deduced from the 2-ary scenario.
\par
According to the works of Davenport \cite{Davenport06detectionand}\cite{SPComp}, Zahedi \cite{Ramin2010} and Rao \cite{RaoICCASP2012}, under the constraint of row-orthogonal ($\bm{\Phi\Phi^T=I})$ to the measurement matrices, the following classifier can be used:
\begin{eqnarray}
t_i&=&\bm{y^T\bm{\Phi} s_i} - \frac{1}{2}\|\bm{\Phi s_i}\|_2^2\nonumber \\
&=&\langle\bm{y},\bm{\Phi}\bm{s_i}\rangle-\frac{1}{2}\langle \bm{\Phi s_i},\bm{\Phi s_i} \rangle,\quad i=1,2 \label{OMPDetct}.
\end{eqnarray}
And the classification result is
\begin{equation}\label{MatchFitler}
i^*=\arg \max_{i} \{t_i\} \quad i=1,2.
\end{equation}
\par
The classifier in (\ref{OMPDetct}) is a most commonly used classifier, it is composed of a term of correlation and an minus term for normalization. The correlation part $\bm{y^T\bm{\Phi} s_i}$ also plays an important role in iterations of the Orthogonal Matching Pursuit (OMP) algorithm \cite{AnalOMP}. Since all the known analysis of Compressive Classifiers is restricted to row-orthogonal measurement matrices, our work is going to break this restriction. So the performance of the Compressive Classifier in (\ref{OMPDetct}) with more general measurement matrices that may or may not be row-orthogonal will be analyzed as follows.
\begin{Theorem}
If we use Compressive Classifier (\ref{OMPDetct}) and (\ref{MatchFitler}) to classify (\ref{CompressHypo}), the probability of false classification can be expressed as:
\begin{equation}
\mathcal{P}_E =
Q(\frac{\|\bm{\Phi(s_1-s_2)}\|_2^2}{2\sigma\|\bm{\Phi^T\Phi(s_1-s_2)}\|_2}) \label{CompressDetect}.
\end{equation}
Here $Q(x)=\int_{x}^{\infty} {\frac{1}{\sqrt{2\pi}}\exp{(-\frac{t^2}{2})}\mathrm{d}t}$, and $\bm \Phi \in \mathbb{R}^{n\times N},n<N$ is the measurement matrix.
\end{Theorem}
\begin{proof}
The probability of false classification, or error probability is \cite{KayV2}
\begin{eqnarray}\label{ErrorProb}
\mathcal{P}_E &=& 1 -
\sum_{k=1}^2 \mathcal{P}(t_k >t_j,j\neq k|H_k)\mathcal{P}(H_k).
\end{eqnarray}\par
In general, the priori probability of different hypotheses can be assumed as uniformly distributed ($\mathcal{P}(H_i)=\mathcal{P}(H_j), i \neq j$), and the conditional probability in (\ref{ErrorProb}) is assumed to be equal \cite{KayV2} for symmetry, that is
\begin{eqnarray}
\mathcal{P}(t_1 > t_2|H_1) =
 \mathcal{P}(t_2 > t_1|H_2)
\end{eqnarray}
then the error probability is
\begin{eqnarray}
\mathcal{P}_E &=& 1 - \mathcal{P}(t_1 > t_2|H_1)
.
\end{eqnarray}
Thus the probability can be analyzed with respect to the conditional probability $\mathcal{P}( t_1> t_2|H_1)$ under the $H_1$ hypothesis.
The statistics (\ref{OMPDetct}) under the $H_1$ hypothesis have the following joint distributions
\begin{eqnarray}
\left[ t_1 \quad t_2 \right]^T
		& \thicksim & \mathcal{N}(\bm \mu_{1,2},\bm \Sigma_{1,2})
\end{eqnarray}
and
\begin{eqnarray}
\bm \mu_{1,2}=\left[ \begin{array}{c}
			 \frac{1}{2}\|\bm{\Phi s_1}\|_2^2 \\
			 \langle\bm{\Phi s_2},\bm{\Phi s_1}\rangle-\frac{1}{2}\|\bm{\Phi s_2}\|_2^2
		\end{array} \right] \nonumber
\end{eqnarray}
\begin{equation}
\bm \Sigma_{1,2} =\sigma^2\left[ \begin{array}{cc}
									\|\bm{\Phi^T\Phi s_1}\|_2^2 &\langle \bm{\Phi^T\Phi s_1},\bm{\Phi^T\Phi s_2}\rangle \\
									\langle \bm{\Phi^1\Phi s_2},\bm{\Phi^T\Phi s_1}\rangle &\|\bm{\Phi^T\Phi s_2}\|_2^2
								\end{array}\right].\nonumber
\end{equation}
Using the property of Gaussian distribution, the probability of false classification is then
\begin{eqnarray}
\mathcal{P}_E 
&=&Q(\frac{\|\bm{\Phi(s_1-s_2)}\|_2^2}{2\sigma\|\bm{\Phi^T\Phi(s_1-s_2)}\|_2})\nonumber
\end{eqnarray}
and $Q(x)=\int_{x}^{\infty} {\frac{1}{\sqrt{2\pi}}\exp{(-\frac{t^2}{2})}\mathrm{d}t}$.
\end{proof}
\par As a matter of fact, if there is more than 2 hypotheses in the Compressive Classification problem (\ref{CompressHypo}), the error probability of the Compressive Classifier (\ref{OMPDetct}) and (\ref{MatchFitler}) may not be so explicit as (\ref{CompressDetect}) for the statistical correlation between different $t_i$'s in (\ref{OMPDetct}). But similar techniques can be utilized and same results can be deduced, we will discuss these m-ary ($m>2$) hypotheses scenarios in the next section.
\par So in order to analyze the error probability (\ref{CompressDetect}) of classifier (\ref{OMPDetct}) without the constraint of row-orthogonality to measurement matrices and for all possible k-sparse signals $\bm s_i,i=1,2$, we will have to focus on
\begin{equation}
\frac{\|\bm{\Phi(s_1-s_2)}\|_2^2}{\|\bm{\Phi^T\Phi(s_1-s_2)}\|_2} \label{MainTarget}
\end{equation}
for all k-sparse signals $\bm{s_1,s_2}\in\Lambda_k=\{\bm{s} \in\mathbb{R}^N,\|\bm{s}\|_0\leq k\}$. \par
In the cases where measurement matrices satisfying row-orthogonality ($\bm{\Phi\Phi^T}=\bm I$), (\ref{MainTarget}) is then reduced to
\begin{equation}
\frac{\|\bm{\Phi(s_1-s_2)}\|_2^2}{\|\bm{\Phi^T\Phi(s_1-s_2)}\|_2} = \|\bm{\Phi(s_1-s_2)}\|_2.
\end{equation}
And this is what Davenport \cite{Davenport06detectionand}\cite{SPComp} and Zahedi \cite{Ramin2010}\cite{Zahedi201264} analyzed in their publications.

\section{measurement Matrices and the Error Probability of Compressive Signal Classification}
\label{sec:main}
\par Although there has been plenty of works about the performance analysis of Compressive Classification, all these works have the same row-orthogonality presumption, but without a theoretical explanation. However, what we believe is that there exist other important reasons for the row-orthogonal condition to be necessary. Here is our main result of this paper:
\begin{Theorem}
In the Compressive Classification problem (\ref{CompressHypo}), by tightening or row-orthogonalizing the measurement matrix $\bm \Phi \in \mathbb{R}^{n \times N},n < N$, the error probability (\ref{CompressDetect}) of the classifier (\ref{OMPDetct}) will be reduced, which means
\begin{equation}\label{MainIneq}
\frac{\|\bm{\Phi(s_1-s_2)}\|_2^2}{\|\bm{\Phi^T\Phi(s_1-s_2)}\|_2} \leq
\frac{\|\bm{\hat \Phi(s_1-s_2)}\|_2^2}{\|\bm{\hat \Phi^T\hat \Phi(s_1-s_2)}\|_2}
\end{equation}
where $\bm \Phi \in \mathbb{R}^{n \times N}$ is the arbitrary measurement matrix, and $\bm {\hat \Phi} \in \mathbb{R}^{n \times N}$ is the equi-norm tight frame measurement matrix row-orthogonalized from $\bm \Phi$.
\end{Theorem}
\begin{proof}
According to Section 2, the error probability (\ref{CompressDetect}) is determined by the following expression (\ref{MainTarget}):
\begin{equation}
\frac{\|\bm{\Phi(s_1-s_2)}\|_2^2}{\|\bm{\Phi^T\Phi(s_1-s_2)}\|_2}  \nonumber
\end{equation}
for all k-sparse signals $\bm s_1, \bm s_2$, where $\bm \Phi$ is an arbitrary measurement matrix satisfying RIP. \par
According to the basic presumptions of $\bm \Phi$ in (\ref{CompressHypo}), the arbitrary under-determined measurement matrix $\bm \Phi \in \mathbb{R}^{n \times N}$, $n < N$ has full row rank, thus the singular value decomposition of $\bm \Phi$ is
\begin{equation}\label{singular}
\bm \Phi = \bm{U \left[ \Sigma_n \quad O\right] V^T}
\end{equation}
Here $\bm \Sigma_n \in \mathbb{R}^{n \times n}$ is a diagonal matrix with each element $\bm \Phi$'s singular value $\sigma_j \neq 0$ $(1\leq j \leq n)$, and $\bm U \in \mathbb{R}^{n \times n}$, $\bm V \in \mathbb{R}^{N \times N}$ are orthogonal matrices composed of $\bm \Phi$'s left and right singular vectors. \par
If an arbitrary equi-norm measurement matrix $\bm \Phi$ is transformed into an equi-norm tight frame $\bm{\hat \Phi}$, we do orthogonalization to its row vectors, which is equivalent as:
\begin{equation}\label{tighten}
\bm{\hat \Phi} = \sqrt{c} \cdot \bm{U\Sigma_n^{-1}U^T\Phi} = \sqrt{c}\cdot \bm{U \left[ I_n \quad O\right] V^T}
\end{equation}
Where $\bm U \in \mathbb{R}^{n \times n}$, $\bm V \in \mathbb{R}^{N \times N}$ are $\bm \Phi$'s singular vector matrices. In a word, row-orthogonalization is equivalent to transforming all singular values of $\bm \Phi$ into equal ones. Thus $\bm{\hat \Phi \hat \Phi^T = c \cdot I_n}$, where $c>0$ is a certain constant for normalization.\par
Then
\begin{eqnarray}\label{TF}
\frac{\|\bm{\hat \Phi(s_1-s_2)}\|_2^2}{\|\bm{\hat \Phi^T\hat \Phi(s_1-s_2)}\|_2}
=
\| \left[ \begin{array}{cc}
									\bm{I_n} &\bm O
								\end{array}\right] \bm V^T(\bm{s_1-s_2})\|_2.
\end{eqnarray}
And for arbitrary measurement matrix $\bm \Phi$ that may not be row-orthogonal, we have
\begin{eqnarray}
\lefteqn{\frac{\|\bm{\Phi(s_1-s_2)}\|_2^2}{\|\bm{\Phi^T\Phi(s_1-s_2)}\|_2}} \nonumber \\
&=& \frac{\|\left[ \begin{array}{cc}
												\bm \Sigma_n & \bm O \end{array} \right]\bm{V^T(s_1-s_2)}\|_2^2}
			{\| \left[ \begin{array}{cc}
												\bm \Sigma_n^2 & \bm O \end{array} \right]\bm{V^T(s_1-s_2)}\|_2} .\label{singularfrac}
\end{eqnarray}
If we denote $\bm{V^T(s_1-s_2)}$ by $\bm u^{(1,2)}$, where $\bm u^{(1,2)} = [u_1, u_2, \cdots ,u_N]^T$. Then (\ref{singularfrac}) becomes
\begin{eqnarray}
\frac{\|\left[ \begin{array}{cc}
												\bm \Sigma_n & \bm O \end{array} \right]\bm{V^T(s_1-s_2)}\|_2^2}
			{\| \left[ \begin{array}{cc}
												\bm \Sigma_n^2 & \bm O\end{array} \right]\bm{V^T(s_1-s_2)}\|_2}
												= \frac{\sum_{j=1}^n {\sigma_j^2 u^{2}_j}}{\sqrt{\sum_{j=1}^n {\sigma_j^4 u^{2}_j}}}\nonumber \\
\leq \sqrt{\sum_{j=1}^n {u^{2}_j}}=
\| \left[ \begin{array}{cc}
									\bm{I_n} &\bm O
								\end{array}\right] \bm V^T(\bm{s_1-s_2})\|_2 . \label{CIneq}
\end{eqnarray}
The last inequality is derived from the Cauchy-Schwarz Inequality, combining (\ref{singularfrac}), (\ref{CIneq}) with (\ref{TF}), then we have
\begin{equation}
\frac{\|\bm{\Phi(s_1-s_2)}\|_2^2}{\|\bm{\Phi^T\Phi(s_1-s_2)}\|_2} \leq
\frac{\|\bm{\hat \Phi(s_1-s_2)}\|_2^2}{\|\bm{\hat \Phi^T\hat \Phi(s_1-s_2)}\|_2}
\end{equation}
which means that row-orthogonalization makes (\ref{MainTarget}) larger and thus brings lower error probability.
The condition when the equality holds is that
\begin{eqnarray}\label{condition}
\bm{\left[\Sigma_n^2\quad O\right]}\bm{V^T}(\bm{s_1-s_2})
=c \cdot \bm{\left[I_n\quad O\right]}\bm{V^T}(\bm{s_1-s_2})
\end{eqnarray}
where $c>0$ is a certain constant.\par
It is obvious that the equality in (\ref{MainIneq}) holds for all k-sparse signals $\bm s_1$ and $\bm s_2$, if and only if $\bm{\Sigma_n^2=c \cdot I_n}$, which means
\begin{equation}
\bm{\Phi\Phi^T} = c \cdot \bm I_n\label{TightCond}
\end{equation}
\par So the result of (\ref{MainIneq}) means that when arbitrary under-determined measurement matrices $\bm \Phi \in \mathbb{R}^{n \times N},n < N$ are transformed into an equi-norm tight frame, i.e. row-ortho-gonalized, the equality in (\ref{MainIneq}) will hold, then the value of (\ref{MainTarget}) will increase, which means improvement of the performance of Compressive Classifier (\ref{OMPDetct}).
\end{proof}\par
The constant $c>0$ above is an amplitude constant for normalization and can take any value, with the equi-norm presumption of measurement matrices, the following corollary can be deduced:
\begin{Corollary}
If a matrix $\bm \Phi\in \mathbb{R}^{n \times N}, n < N$ form an equi-norm tight frame, that is $\bm{\Phi\Phi^T}=c\cdot I_n$ and the column vectors satisfy $\|\bm \phi_i\|_2=\|\bm \phi_j\|_2=\psi$, then $c = \frac{N}{n}\psi^2$.
\end{Corollary}
\begin{proof}
If $\bm \Phi$ has equal column norms and satisfies $\bm{\Phi\Phi^T}=c\cdot I_n$, then
\begin{eqnarray}
tr(\bm{\Phi^T\Phi})=N\cdot \psi ^2= tr(\bm{\Phi\Phi^T})=n\cdot c.
\end{eqnarray}
As a result, $c=\frac{N}{n}\psi^2$.
\end{proof}\par
If we let $c = 1$, then we can get $\|\bm \phi_i\|_2=\|\bm \phi_j\|_2=\sqrt{n/N}$, which coincides with the results of \cite{Ramin2010} and \cite{Zahedi201264}.

\par Before the end of this section, some important discussions are believed to be necessary here.
\par Remark 1:
Further analysis of the result of Theorem 2 indicates that, when the measurement matrices of the commonly used Compressive Classifier (\ref{OMPDetct}) are "tightened", i.e. row-orthogonalized, then the inequality (\ref{MainIneq}) becomes equality, and the corresponding error probability will become
\begin{eqnarray}\label{MFeqivalent}
\mathcal{P}_E(\hat{\bm \Phi})=
 Q(\frac{\|\bm{\hat \Phi(s_1-s_2)}\|_2}{2 c^{-1/2}\cdot \sigma})\nonumber \\
= Q(\frac{\|\bm{P_{\Phi^T} (s_1-s_2)}\|_2}{2 c^{-1/2}\cdot \sigma})
\end{eqnarray}
where $\bm{P_{\Phi^T} = \Phi^T ( \Phi  \Phi^T)^{-1} \Phi}$. The last equality is derived from (\ref{singular}) and (\ref{tighten}), which equals, as a matter of fact, the error probability of the General Matched Filter Classifier \cite{SPComp}\cite{PHDThesis}:
\begin{eqnarray}
\hat t_i=
\langle\bm{y},\bm{(\Phi \Phi^T)^{-1}\Phi}\bm{s_i}\rangle-\frac{1}{2}\|\bm{P_{\Phi^T} s_i} \|_2^2,\quad i=1,2. \label{MFDetct}
\end{eqnarray}
Thus Theorem 2 and (\ref{MFeqivalent}) indicates that equi-norm tight frames will improve the Compressive Classifier (\ref{OMPDetct}) to the level of General Matched Filter Classifier (\ref{MFDetct}) in the sense of error probability. Although it is obvious that row-orthogonality (\ref{TightCond}) sufficiently ensures (\ref{MFDetct}) to become equivalent as (\ref{OMPDetct}), the necessity with row-orthogonality, or "tightness", of the equivalence between the Compressive Classifier (\ref{OMPDetct}) and the General Matched Filter Classifier (\ref{MFDetct}) is not so explicit but demonstrated by Theorem 2 and (\ref{MFeqivalent}). 
\par Besides, the error probability (\ref{MFeqivalent}) coincides with Davenport's \cite{Davenport06detectionand} and \cite{SPComp} and Zahedi's \cite{Ramin2010} and \cite{Zahedi201264}, where they constrained row-orthogonality $\bm{\Phi \Phi^T = c \cdot I}$ and set $c = 1$. So Theorem 2 explains the benefits of using the row-orthogonal constraint to do Compressive Classification.
 Similar discussions about the improvement of equi-norm tight frames to oracle estimators can be found in \cite{UniTightFrame}, which is another good support of the advantage of "tight".
\par Remark 2:
As is mentioned in last section, when there are more than 2 hypotheses, the m-ary ($m>2$) compressive classification problem model becomes:
\begin{equation}\label{CompressHypo_mary}
\bm{y} =
\left \{ \begin{array}{ll}
	\bm{\Phi}(\bm{s_1} + \bm{n}) & \text{Hypothesis  $H_1$} \\
	\bm{\Phi}(\bm{s_2} + \bm{n}) & \text{Hypothesis  $H_2$} \\
	\cdots & \cdots \\
	\bm{\Phi}(\bm{s_m} + \bm{n}) & \text{Hypothesis  $H_m$} \\

\end{array} \right. .
\end{equation}
Using the same Compressive Classifier
\begin{eqnarray}
t_i
=\langle\bm{y},\bm{\Phi}\bm{s_i}\rangle-\frac{1}{2}\langle \bm{\Phi s_i},\bm{\Phi s_i} \rangle,\quad i=1,2,\cdots ,m .
\end{eqnarray}
The corresponding error probability will be 
\begin{eqnarray}
\mathcal{P}_E 
=1 - \mathcal{P}( t_T> t_i,\forall i \neq T|H_T) ,\quad T = 1,2,\cdots,m.
\end{eqnarray}
Combined with the Union Bound of probability theory, the error probability then satisfies
\begin{eqnarray}
\mathcal{P}_E \leq
\sum_{i \neq T}^m Q(\frac{\|\bm{\Phi(s_T-s_i)}\|_2^2}{2\sigma\|\bm{\Phi^T\Phi(s_T-s_i)}\|_2}), T = 1,2,\cdots,m .\label{mErrorProb}
\end{eqnarray}
The error probability (\ref{mErrorProb}) is similar to (\ref{CompressDetect}) except for the inequality due to the use of union bound. In fact, it may be difficult to get any more accurate result than (\ref{mErrorProb}), because of the statistical correlation between different $t_i$'s. It may not be persuasive to conclude the error probability's decrease brought by equi-norm tight frames, using the same proof in Theorem 2 in this m-ary scenario. because of the inequality in (\ref{mErrorProb}); however, simulation results in the next section will demonstrate that equi-norm tight frames are still better in the m-mary ($m>2$) Compressive Classification scenario.
\par Remark 3: In comparison with the the work of Zahedi in \cite{Ramin2010} and \cite{Zahedi201264}, where Equiangular Tight Frames (ETFs) are proved to have the best worst-case performance among all tight frames (row-orthogonal constrained matrices), we just give the proof that for general under-determined measurement matrices, tightening can bring performance improvement for Compressive Classification. Our job is different from theirs, because all of Zahedi's analysis is based on the constraint that the measurement matrices are tight, or row-orthogonal, and the advantage of Equiangular Tight Frames (ETFs, \cite{Waldron2009}) is that ETFs have the best worst-case (maximum of the minimum) performance among all tight frames of same dimensions, while our result shows that when arbitrary measurement matrices is "tightened", i.e. transformed into equi-norm tight frames, the performance of Compressive Classification will get improved. Nonetheless, the existence and construction of ETFs of some certain dimensions remains an open problem (\cite{Waldron2009}), while doing row-orthogonalization for arbitrary matrices is very easy and practical. So our results provided a convenient approach to improve the performance of compressive classifiers.
\par

\section{Simulations}
\label{sec:simulation}

\begin{figure}
\centering
	\begin{minipage}[htbp]{0.5\textwidth}
	\centering
		\includegraphics[width=3.6in]{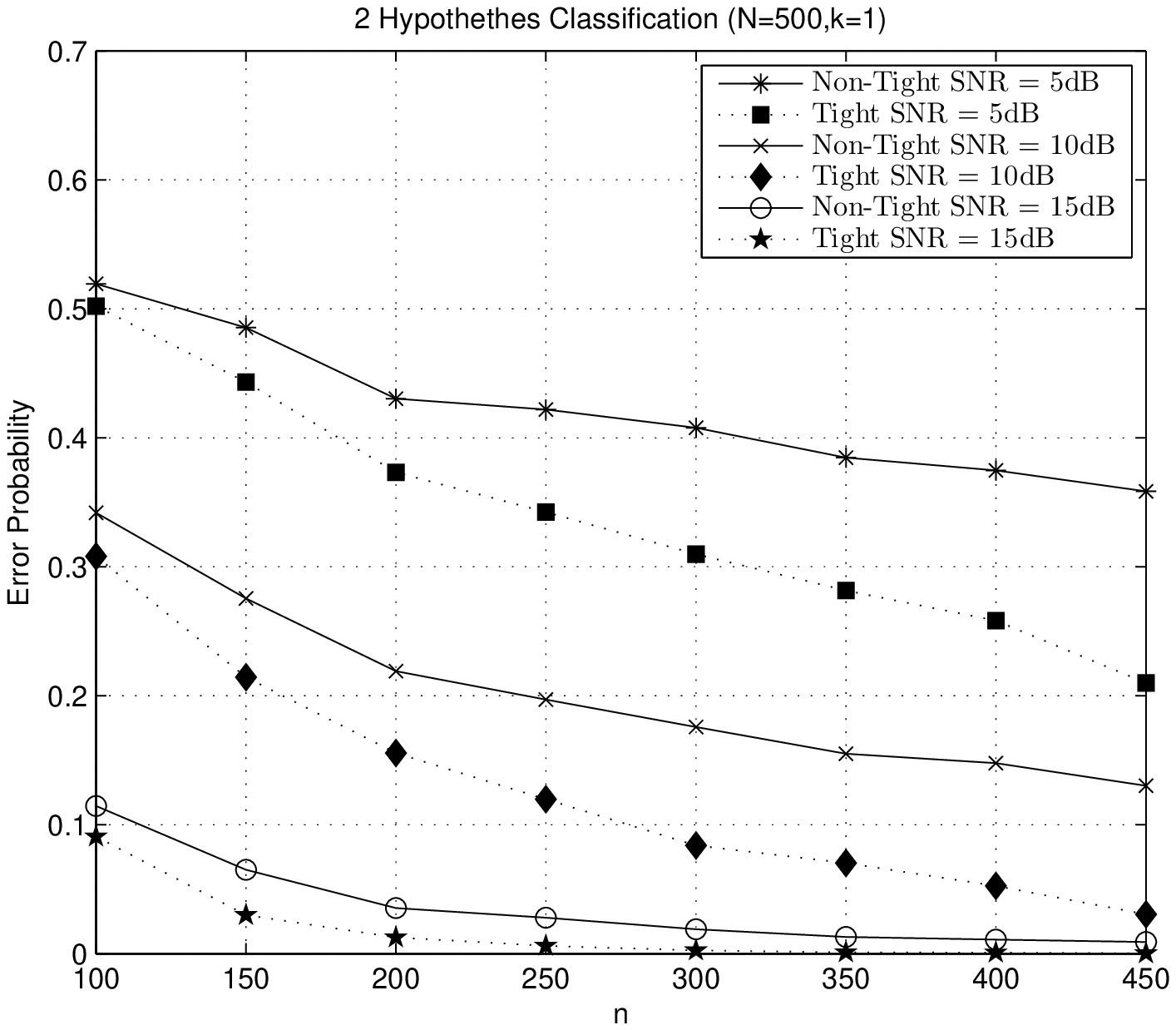}
		\caption{Monte-Carlo simulation of 2-ary compressive classification error probability using non-tight frames and tight frames (for $k=1$ sparse signals)}
		\label{figure1}
		\end{minipage}
	\begin{minipage}[htbp]{0.5\textwidth}
	\centering
		\includegraphics[width=3.6in]{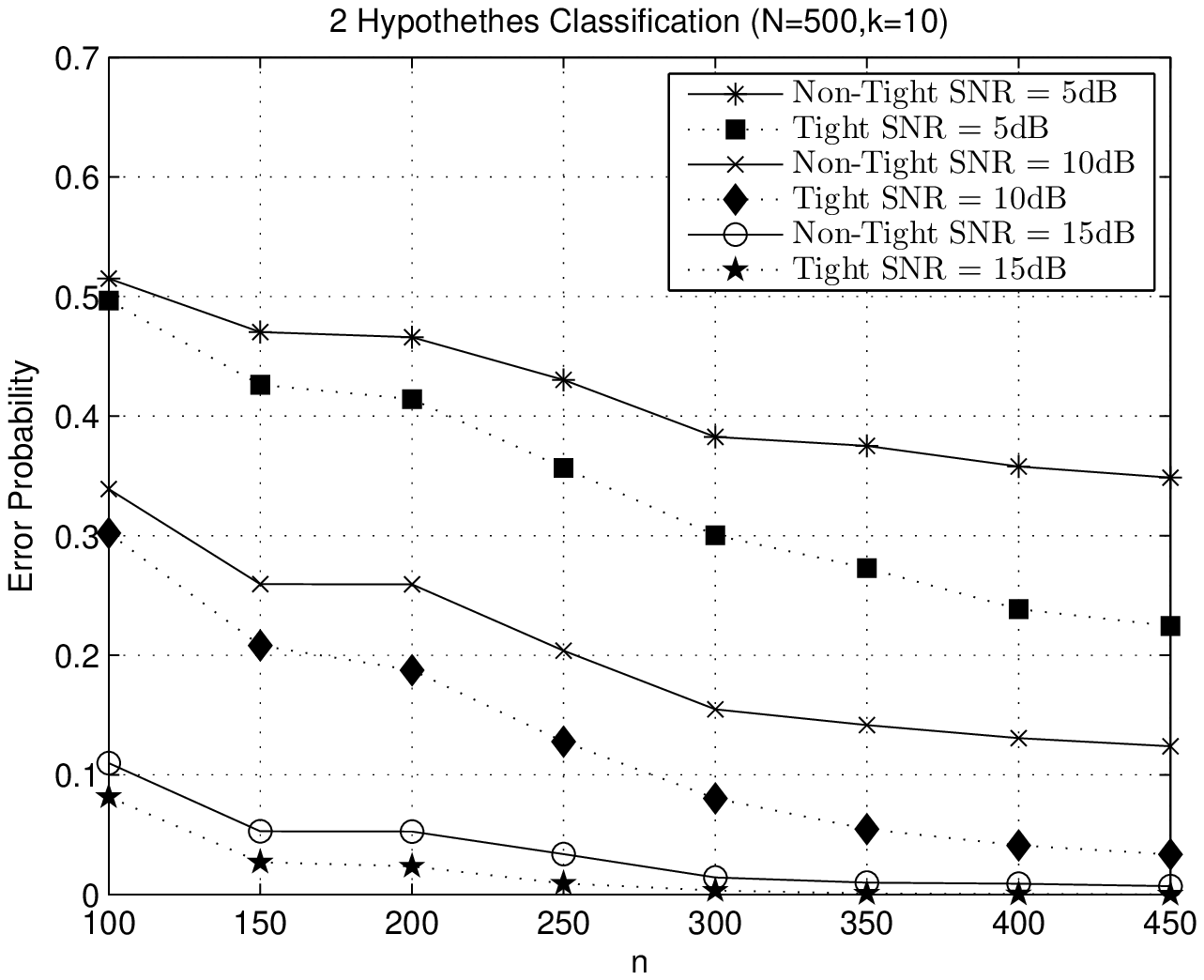}
		\caption{Monte-Carlo simulation of 2-ary compressive classification error probability using non-tight frames and tight frames (for $k=10$ sparse signals)}
		\label{figure2}
		\end{minipage}
	
	\end{figure}
	\begin{figure}
		\begin{minipage}[htbp]{0.5\textwidth}
		\centering
		\includegraphics[width=3.6in]{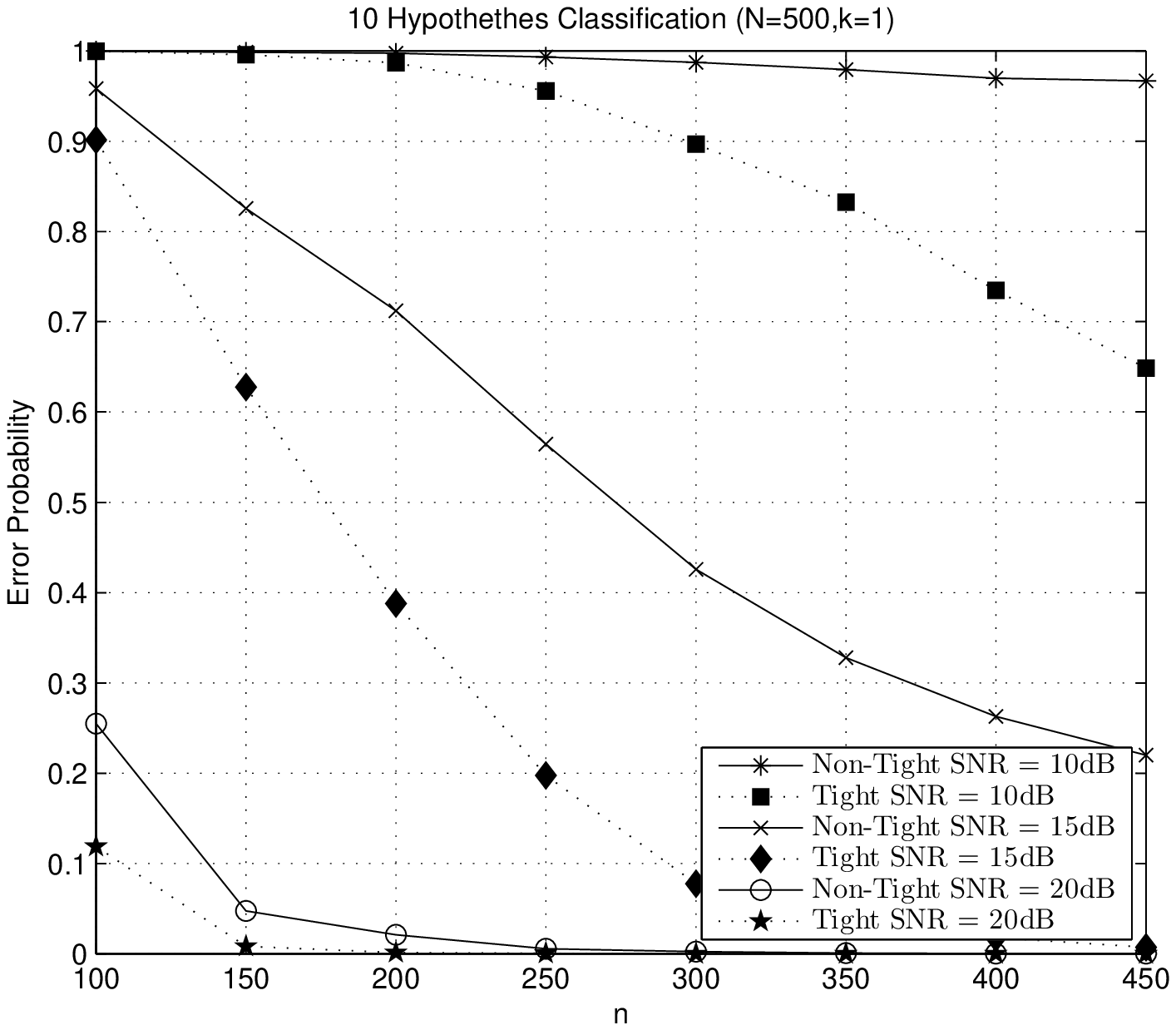}
		\caption{Monte-Carlo simulation of 10-ary compressive classification error probability using non-tight frames and tight frames (for $k=1$ sparse signals)}
		\label{figure3}		
	\end{minipage}\\ %
		\begin{minipage}[htbp]{0.5\textwidth}
		\centering
		\includegraphics[width=3.6in]{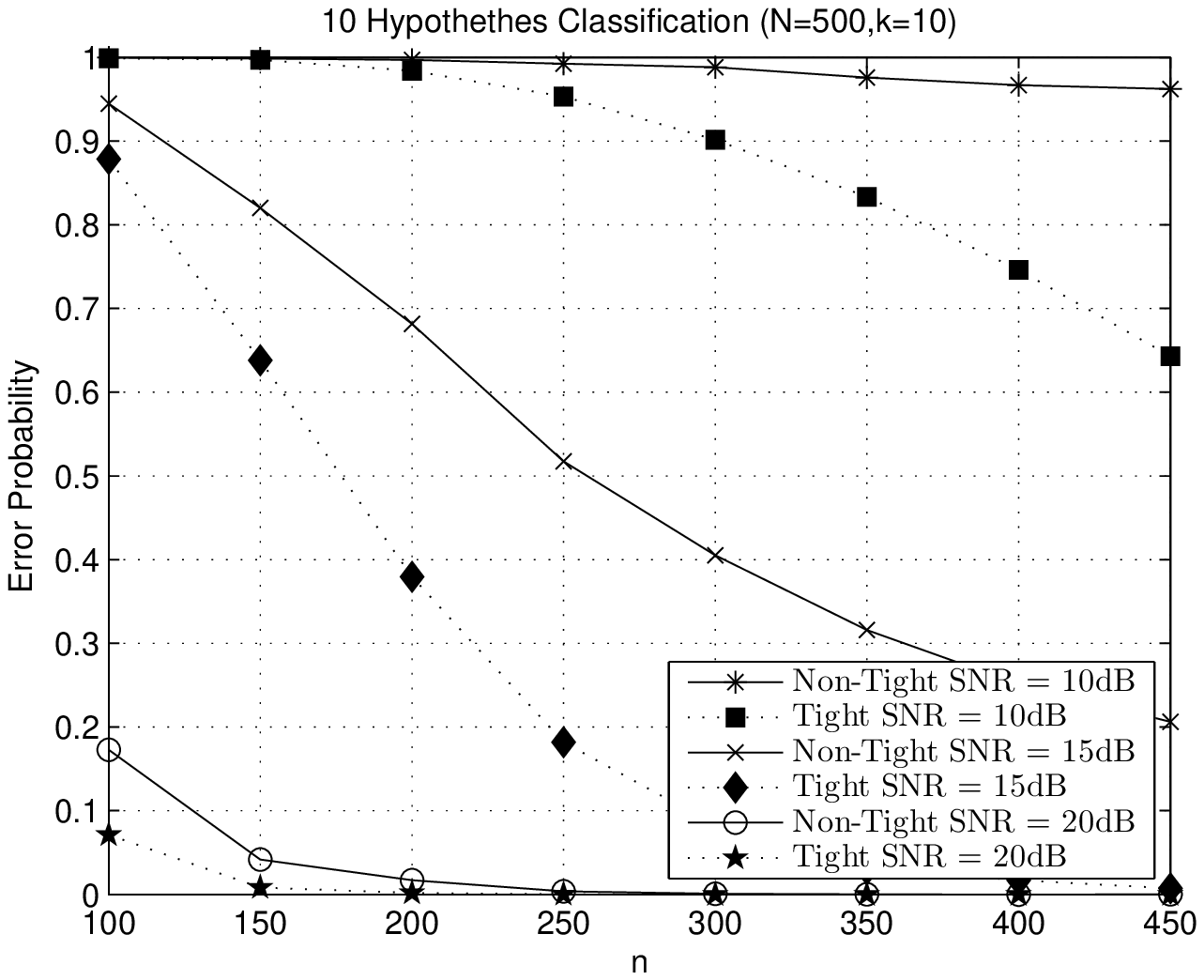}
		\caption{Monte-Carlo simulation of 10-ary compressive classification error probability using non-tight frames and tight frames (for $k=10$ sparse signals)}
		\label{figure4}		
	\end{minipage}%
\end{figure}
\par In this section the main result of theorem 2 is verified by Monte-Carlo simulations. In the simulation some arbitrary $k=1$ and $k=10$ sparse signals are generated and classified using non-tight frames and tight frames. The Gaussian Random Matrices are chosen to be the non-tight frames, and the row-orthogonalized ones from those random matrices are chosen as the tight frames. Here we choose $N=500$, and the error probabilities of both 2-ary Compressive Classification and 10-ary Compressive Classification are demonstrated in Fig.\ref{figure1} and Fig.\ref{figure3} for $k=1$ sparse signals, and Fig.\ref{figure2} and Fig.\ref{figure4} for $k=10$ sparse signals, with the number of measurements $n$ ranging from 100 to 450 and signal to noise ratios $\bm{\|s_i\|_2^2}/\sigma^2$ from 5 dB to 20 dB. Each error probability is calculated from average of 10000 independent experiments with tight or non-tight measurement matrices.
\par The simulation shows that equi-norm tight frames transformed from general Gaussian Random Matrices have better Compressive Classification performance than those non-tight Gaussian Random Matrices within $n$'s whole range, both for 2-ary classification and m-ary ($m>2$) classification scenarios, which is the benefit that "tightening" brings.

\section{Conclusion}
\label{sec:conclusion}
\par This paper deals with the performance improvement of a commonly used Compressive Classifier (\ref{OMPDetct}). We prove that the transformation to equi-norm tight Frames from arbitrary measurement matrices will reduce the probability of false classification of the commonly used Compressive Classifier, thus improve the classification performance to the level of the General Matched Filter Classifier (\ref{MFDetct}), which coincides with the row-orthogonal constraint commonly used before. Although there are other proofs that among all equi-norm tight frames the Equiangular Tight Frames (ETFs) achieve best worst-case classification performance, the existence and construction of ETFs of some dimensions is still an open problem.
As the construction of equi-norm tight frames from arbitrary matrices is much simple and practical, the conclusion of this paper can also provide a convenient approach to implement an improved measurement matrix for Compressive Classification.



\bibliographystyle{IEEEbib}
\bibliography{CompressiveDetection.bib}

\end{document}